\documentclass[letterpaper, 10pt, conference]{ieeeconf}
\IEEEoverridecommandlockouts
\usepackage{geometry}
\usepackage{graphicx}
\usepackage{epsfig}
\usepackage{amsmath}
\usepackage{amssymb} 
\usepackage[noadjust]{cite}
\usepackage{bm}
\usepackage{subfigure}
\usepackage{acronym}
\usepackage{paralist}
\usepackage{float}
\usepackage{epstopdf}
\usepackage{algorithm}
\usepackage{algpseudocode}
\usepackage{mathtools}
\usepackage{multicol} 
\usepackage{accents}
\usepackage{centernot}
\usepackage{xcolor}

\makeatletter
\def\BState{\State\hskip-\ALG@thistlm}
\makeatother

\newcommand{\sfn}{S_f^\N}

\newtheorem{problem}{Problem}

\newtheorem{define}{Definition}
\newtheorem{theorem}{Theorem}
\newtheorem{lemma}{Lemma}
\newtheorem{remark}{Remark}
\newtheorem{assume}{Assumption}

\newcommand{\R}{\mathbb R}

\newcommand{\N}{\mathcal{N}}
\newcommand{\D}{\mathcal{D}}
\newcommand{\V}{\mathcal{V}}

\newcommand{\E}{\mathcal{E}}

\newcommand{\Le}{\mathcal{L}}
\newcommand{\A}{\mathcal{A}}
\newcommand{\J}{\mathcal{J}}
\newcommand{\eps}{\epsilon}

\newcommand{\Z}{\mathbb{Z}}

\newcommand{\Tc}{\mathcal{T}}

\newcommand{\pth}[1]{\left(#1\right)} 

\DeclarePairedDelimiter{\ceil}{\lceil}{\rceil}
\DeclarePairedDelimiter{\floor}{\lfloor}{\rfloor}
\DeclarePairedDelimiter{\abs}{\lvert}{\rvert}

\newcommand{\rarr}{\rightarrow} 
 
\makeatletter
\let\oldceil\ceil
\def\ceil{\@ifstar{\oldceil}{\oldceil*}}
\makeatother
\makeatletter
\let\oldfloor\floor
\def\floor{\@ifstar{\oldfloor}{\oldfloor*}}
\makeatother
\makeatletter
\let\oldnorm\norm
\def\norm{\@ifstar{\oldnorm}{\oldnorm*}}
\makeatother
\makeatletter
\let\oldabs\abs
\def\abs{\@ifstar{\oldabs}{\oldabs*}}
\makeatother

\begin{document}

\title{\LARGE \bf Resilient Leader-Follower Consensus with Time-Varying Leaders in Discrete-Time Systems}

\author{James Usevitch and Dimitra Panagou
\thanks{James Usevitch is with the Department of Aerospace Engineering, University of Michigan, Ann Arbor; \texttt{usevitch@umich.edu}.
Dimitra Panagou is with the Department of Aerospace Engineering, University of Michigan, Ann Arbor; \texttt{dpanagou@umich.edu}.
The authors would like to acknowledge the support of the Automotive Research Center (ARC) in accordance with Cooperative Agreement W56HZV-14-2-0001 U.S. Army TARDEC in Warren, MI, and of the Award No W911NF-17-1-0526.}
}

\maketitle
\thispagestyle{empty}
\pagestyle{empty}

\acrodef{wrt}[w.r.t.]{with respect to}
\acrodef{apf}[APF]{Artificial Potential Fields}
\begin{abstract}
	The problem of consensus in the presence of adversarially behaving agents has been studied extensively in the literature. The proposed  algorithms typically guarantee that the consensus value lies within the convex hull of initial normal agents' states. In leader-follower consensus problems however, the objective for normally behaving agents is to track a time-varying reference state that may take on values outside of this convex hull. In this paper we present a method for agents with discrete-time dynamics to resiliently track a set of leaders' common time-varying reference state despite a bounded subset of the leaders and followers behaving adversarially. The efficacy of our results are demonstrated through simulations.
\end{abstract}

\IEEEpeerreviewmaketitle

\section{Introduction}

{

Guaranteeing the resilience of multi-agent systems against adversarial misbehavior and misinformation is a critically needed property in modern autonomous systems. As part of this need, the \emph{resilient consensus problem} has been the focus of much research for the past few decades. In this problem, normally behaving agents in a multi-agent network seek to come to agreement on one or more state variables in the presence of adversarially behaving agents whose identity is unknown. Within the last decade, several algorithms based upon the \emph{Mean-Subsequence-Reduced} family of algorithms \cite{kieckhafer1994reaching} have become a popular means to guarantee consensus of the normally behaving agents when the number of adversaries is bounded and the network communication structure satisfies certain \emph{robustness} properties. These discrete-time algorithms, which include the W-MSR, DP-MSR, SW-MSR, and QW-MSR algorithms \cite{leblanc2013resilient,dibaji2017resilient,saldana2017resilient,dibaji2018resilient}, guarantee that the final consensus value of the normal agents is within the convex hull of the normal agents' initial states.

A related problem in prior control literature is the \emph{leader-follower consensus problem}, where the objective is for normally behaving agents to come to agreement on the (possibly time-varying) state of a leader or set of leaders \cite{Dimarogonas2009leader,Ren2007,Ren2008consensus}. Prior work in this area typically assumes that there are no adversarially misbehaving agents; i.e. all leaders and followers behave normally. An interesting question to consider is whether the property of resilience can be extended to the leader-follower consensus scenario, i.e. whether follower agents can track the leader agents' states while rejecting the influence of adversarial agents whose identity is unknown. One aspect which prevents prior resilient consensus results from being extended to this case is that the (possible time-varying) leaders' states may not lie within the convex hull of the initial states of normal agents.
}

Recent work most closely related to the resilient leader-follower consensus problem includes \cite{leblanc2014,Mitra2016secure,mitra2018resilient,usevitch2018resilient}. In \cite{leblanc2014}, the problem of resilient distributed estimation is considered where agents employ a discrete-time resilient consensus algorithm to reduce the estimation error of individual parameters of interest. The authors assume that certain agents have a precise knowledge of their own parameters. These "reliable agents" drive the errors of the remaining normal agents to the static reference value of zero in the presence of misbehaving agents. In \cite{Mitra2016secure,mitra2018secure}, the problem of distributed, resilient estimation in the presence of misbehaving nodes is treated. The authors show conditions under which information about the decoupled modes of the system is resiliently transmitted from a group of source nodes to other nodes that cannot observe those modes. Their results guarantee exponential convergence to the reference modes of the system. In our prior work \cite{usevitch2018resilient}, we considered the case of leader-follower consensus to arbitrary static reference values using the W-MSR algorithm \cite{leblanc2013resilient}. In addition, the resilient leader-follower consensus problem is closely related to the \emph{secure broadcast} problem \cite{koo2006reliable,litsas2013graph,bonomi2018reliable}, where a dealer agent seeks to broadcast a message to an entire network in the presence of misbehaving agents.

{
In this paper, we address the problem of resilient leader-follower consensus with \emph{time-varying leaders} in the discrete-time domain. Specifically, we make the following contributions:
\begin{itemize}
	\item We introduce the Multi-Source Resilient Propagation Algorithm (MS-RPA). Under this algorithm, exact tracking of a set of time-varying leaders can be achieved by normally behaving follower agents in a finite time, despite the presence of a bounded number of arbitrarily misbehaving follower \emph{and leader} agents.
	\item We demonstrate conditions under which agents applying the MS-RPA algorithm achieve exact tracking of leader agents in finite time when agents are subject to input bounds while in the presence of a bounded number of arbitrarily misbehaving follower and leader agents.
\end{itemize}
Notation and relevant resilient consensus terms from prior literature are reviewed in Section \ref{sec:notation}. The problem formulation is given in Section \ref{sec:problem}. Our main results are outlined in Section \ref{sec:mainresults}. Simulations demonstrating our results are shown in Section \ref{sec:simulations}, and a brief conclusion is given in Section \ref{sec:conclusion}.

}

\section{Notation and Preliminaries}
\label{sec:notation}

{
The set of real numbers and integers are denoted $\R$ and $\Z$, respectively. The set of nonnegative reals and nonnegative integers are denoted $\R_+$ and $\Z_+$, respectively.
The cardinality of a set $S$ is denoted $|S|$. The set union, intersection, and set difference operations of two sets $S_1$ and $S_2$ are denoted by $S_1 \cup S_2$, $S_1 \cap S_2$, and $S_1 \backslash S_2$ respectively. We denote $\bigcup_{i=1}^n S_i = S_1 \cup S_2 \cup \ldots \cup S_n$. We also denote $B(x,r) = \{z \in \R : |x - z| \leq r \}$.
A digraph is denoted as $\D = (\V,\E)$ where $\V$ is the set of vertices or agents, and $\E \subset \V \times \V$ is the set of edges. An edge from $i$ to $j$, $i,j \in \V$, denoted as $(i,j) \in \E$, represents the ability of the \emph{head} $i$ to send information to the \emph{tail} $j$. Note that for digraphs $(i,j) \in \E \centernot\implies (j,i) \in \E$. The set of in-neighbors of agent $i$ is denoted $\V_i = \{j \in \V: (j,i) \in \E)\}$. Similar to \cite{leblanc2013resilient}, we define the inclusive neighbors of node $i$ as $\J_i = \V_i \cup \{i\}$. The set of out-neighbors of each agent $i$ is denoted $\V_i^{out} = \{k \in \V : (i,k) \in \E \}$.
}

\subsection{Resilient Consensus Preliminaries}

{
We briefly review several definitions associated with the resilient consensus literature that will be used in this paper.

\begin{define}[\cite{leblanc2013resilient}]
Let $F \in \Z_+$. A set $S \subset \V$ is \emph{F-total} if it contains at most $F$ nodes; i.e. $|S| \leq F$.
\end{define}

\begin{define}[\cite{leblanc2013resilient}]
Let $F \in \Z_+$. A set $S \subset \V$ is \emph{F-local} if $\forall t \geq t_0$, $t \in \Z$, $|S \cap \V_i(t)| \leq F$ $\forall i \in \V \backslash S$. 
\end{define}
In words, a set $S$ is $F$-local if it contains at most $F$ nodes in the neighborhood of each node outside of $S$ for all $t \geq t_0$.

}

The notion of strong $r$-robustness with respect to (w.r.t.) a subset $S$ is defined as follows:
\begin{define}[Strong $r$-robustness w.r.t. $S$ \cite{Mitra2016secure}]
Let $r \in \Z_+$, $\D=(\V,\E)$ be a digraph, and $S \subset \V$ be a nonempty subset. $\D$ is strongly $r$-robust w.r.t. $S$ if for any nonempty $C \subseteq \V \backslash S$, $C$ is $r$-reachable.
\end{define}

\begin{remark}
	Given a subset $S \subset \V$, it can be verified in polynomial time whether $\D$ is strongly robust w.r.t. $S$ \cite{mitra2018resilient}. 
\end{remark}

\section{Problem Formulation}
\label{sec:problem}

{

We consider a discrete-time network of $n$ agents with the set of agents denoted $\V$. The communication links between the agents are modeled by the digraph $\D = (\V,\E)$. Each agent $i \in \V$ has state $x_i(t) \in \R$. In addition, each agent is designated to be either a leader agent or a follower agent.
\begin{define}
	The set of leader agents is denoted $\Le \subset \V$. The set of follower agents is denoted $S_f = \V \backslash \Le$.
\end{define}

\begin{assume}
\label{a:partition}
The sets $\Le$ and $S_f$ are static and satisfy $\Le \cup S_f = \V$ and $\Le \cap S_f = \emptyset$.
\end{assume}

In many leader-follower settings, the state of the leaders is nominally designated to update according to a desired reference signal function which is propagated to the followers. We consider the case where followers do not have direct access to the reference signal function, but must rely on a feedforward signal communicated by the leaders and relayed through the network of followers. 

More formally, each agent has the communication rate $\eta \in \Z_+$, where $\eta$ represents the number of discrete communication time steps in between each state update time step. 
The set of time instants $\Tc = \{\tau \eta : \tau \geq 1,\ \tau \in \Z_+\}$ represent the time instances where agents update their states.
The set of discrete time steps $t \in \Z_+$ represent the communication times of each agent in $\V$. At each time step  $t \in \Z_+$, each agent $i \in \V$ is able to send its state value $x_i(t)$ to each  its out-neighbors $k \in \V_i^{out}$ and receive state values $x_k(t)$ from each of its in-neighbors $j \in \V_i$.

Each normally behaving leader agent $l$ has the dynamics
\begin{align}
	\label{eq:leader}
	x_l(t + 1) = \begin{cases} f_r(\tau) & \text{if } t - t_0 = \tau \eta -1,\ \tau \geq 1,\ \tau \in \Z_+ \\  x_l(t) & \text{otherwise},  \end{cases}
\end{align}
where $f_r : \R \rarr \R$ is the reference signal.
Each normally behaving follower agent $j$ has dynamics
\begin{align}
	\label{eq:follower}
	x_j(t+1) = x_j(t) + \alpha(t) u_j(t)
\end{align}
where 
\begin{align}
\label{eq:alpha}
	\alpha(t) = \begin{cases}1 & \text{if } t - t_0 = \tau \eta -1,\ \tau \geq 1,\
	\tau \in \Z_+ \\ 0 & \text{otherwise} \end{cases}
\end{align}
the control input functions $u_j : \R \rarr \R$ will be defined later. 


In contrast to much of the prior literature which typically assumes all agents apply nominally specified control laws, this paper considers the presence of \emph{misbehaving agents}:
\begin{define}
\label{def:misbehaving}
    An agent $j \in \V$ is \emph{misbehaving} if at least one of the following conditions hold:
    \begin{enumerate}
        \item There exists a time $t$ where agent $j$ does not update its state according to \eqref{eq:leader} and also does not update its state according to \eqref{eq:follower}.
        \item There exists a time $t$ where $j$ does not communicate its true state value $x_j(t)$ to at least one of its out-neighbors.
        \item There exists a time $t$ where $j$ communicates different state values to different out-neighbors.
\end{enumerate}
The set of misbehaving agents is denoted $\A \subset \V$.
\end{define}
\begin{define}
	The set of agents which are \emph{not} misbehaving are denoted $\N = \V \backslash \A$. Agents in $\N$ are referred to as \emph{normally behaving agents}.
\end{define}
Intuitively, misbehaving agents are agents which update their states arbitrarily or communicate false information to their out-neighbors. By Definition \ref{def:misbehaving}, the set of misbehaving agents $\A$ includes faulty agents, malicious agents, and Byzantine agents \cite{leblanc2013resilient}.

This paper considers scenarios where both followers \emph{and leaders} are vulnerable to adversarial attacks and faults, and therefore the set $\A \cap \Le$ may possibly be nonempty, and the set $\A \cap S_f$ may possibly be nonempty. The following notation will be used:
\begin{define}[Normally behaving agent notation]
The set of normally behaving leaders is denoted as $\Le^\N = \Le \backslash \A$. The set of normally behaving followers is denoted $S_f^\N = S_f \backslash \A$. 
\end{define}
\begin{define}[Misbehaving agent notation]
The set of misbehaving leaders is denoted as $\Le^\A = \Le \cap \A$. The set of misbehaving followers is denoted as $S_f^\A = S_f \cap \A$.
\end{define}

The purpose of this paper is to determine conditions under which normally behaving follower agents resiliently achieve consensus with the time-varying reference state of the set of normally behaving leader agents in the presence of a possibly nonempty set of misbehaving agents.

\begin{problem}[Resilient Leader-Follower Consensus]
\label{prob:leadfollow}
Define the error function $e(t) = \max_{i \in \sfn, l \in \Le^\N} |x_i(t) - x_l(t)|$. Determine conditions under which $e(t)$ is nonincreasing for all $t \geq t_0$, and there exists a finite $T \in \Z_+$ such that $e(t) = 0$ for all $t \geq t_0 + T$, in the presence of a possibly nonempty adversarial set $\A$.

\end{problem}

\section{Main Results}
\label{sec:mainresults}

By Definition \ref{def:misbehaving}, misbehaving agents may either update their state arbitrarily, or communicate misinformation to their out-neighbors. This brings up several challenges for normal followers seeking to track the leaders' reference state:
\begin{enumerate}
	\item Agents in $S_f^\N$ may have no knowledge about the set $\Le$. It is not assumed that leader agents are identifiable to any follower agent.	
	\item The set of misbehaving leaders $\Le^\A$ may possibly be nonempty, implying some normally behaving agents may receive misinformation from the leader set $\Le$. 
	\item There may exist agents $i \in S_f^\N$ which are not directly connected to any leaders, i.e. $\V_i \cap \Le = \emptyset$.
\end{enumerate}
To address these challenges, we introduce the \emph{Multi-Source Resilient Propagation Algorithm} with parameter $F$, which is described in Algorithm \ref{alg:MSRPA}.
\begin{algorithm}
\caption{\small{\textsc{MS-RPA with parameter $F$}:}}
\label{alg:MSRPA}

\textbf{Initialization:}
	\begin{enumerate}
		\item Each leader $l \in \Le$ begins with $x_l(t_0) = f_r(0)$
		\item Each follower $i \in S_f$ begins with $x_i(t_0) \in \R$ and $u_i(t_0) = 0$.
	\end{enumerate}
\textbf{At each time step $\bm{t \geq t_0,\ t \in \Z_+}$:}

	\emph{(Let $\tau'= \floor{\frac{t-t_0}{\eta} + 1}$, where $\tau' \eta$ is the next state update time step and $(\tau'-1)\eta$ is the most recent state update time step.)}
		\begin{enumerate}
	\item Each leader $l \in \Le$ updates its state as per \eqref{eq:leader}
	\item Each follower $i \in S_f$ updates its state as per \eqref{eq:follower}
	\item Each leader agent $l \in \Le$ broadcasts the value $f_r(\tau')$ to its out-neighbors.
	\item If $t \geq (\tau'-1)\eta + 1$, each follower agent $i$ which received a value $c \in \R$ from at least $F+1$ in-neighbors at time step $t-1$ sets 
	\begin{align*}
		u_i(t) = \begin{cases} c - x_i(t) & \text{if input unbounded} \\
		\frac{(c - x_i(t))u_M}{\max(u_M,|c - x_i(t)|)} & \text{if input bounded by } u_M, \end{cases}
	\end{align*}
	where $u_M \in \R$, $u_M > 0$. Denote this set of follower agents as $\mathcal{C}(t)$.
	\item Each agent $i \in \mathcal{C}$ broadcasts $c$ to its out-neighbors for all timesteps in the interval $[t,\tau'\eta-1]$.
	\item Each agent $i \notin \mathcal{C}(t)$ sets $u_i(t) = u_i(t-1)$.
\end{enumerate}
\end{algorithm}
\begin{remark}
	The MS-RPA algorithm is based on the Certified Propagation Algorithm (CPA) in \cite{koo2006reliable} (see also \cite{Zhang2012robustness}), but contains several significant theoretical differences. First, the CPA algorithm considers a single source agent invulnerable to adversarial attacks propagating a reference value to the network, whereas the MS-RPA considers multiple sources that each are \emph{vulnerable} to attacks (i.e. the set $\Le$). Second, the CPA algorithm considers the broadcast of one static message to the network, whereas the MS-RPA considers a time-varying reference signal. Third, the CPA algorithm does not take any state update input bounds into account, whereas the MS-RPA algorithm explicitly considers input bounds.
\end{remark}

Our first Lemma demonstrates that on every time interval $[t_0 + (\tau-1)\eta,t_0 + \tau\eta)$ for $\tau \geq 1$, the MS-RPA algorithm guarantees that the value of $f_r(\tau\eta)$ propagates resiliently from $\Le$ to all normal followers when $\eta$ is sufficiently large.

\begin{lemma}
\label{lem:broadcast}
Let $\D = (\V,\E)$ be a nonempty, nontrivial, simple digraph with $S_f$ nonempty. Let $F,\tau \in \Z_+$, $t_0 \in \Z$, and $u_M \in \R$, $u_M > 0$.
Suppose that $\A$ is an $F$-local set, $\D$ is strongly $(2F+1)$-robust w.r.t. the set $\Le$, and all normally behaving agents apply the MS-RPA algorithm with parameter $F$. If $\eta > |S_f|$, then for all $i \in \sfn$ and $\forall \tau \geq 1$ the control input $u_i$ satisfies 
{
\medmuskip=0mu
\thinmuskip=0mu
\thickmuskip=0mu
\begin{align}
u_i(t_0 + \tau\eta - 1) = \begin{cases}
f_r(\tau) - x_i(t_0 + \tau\eta - 1) & \parbox{5em}{if input unbounded,} \\
\frac{(f_r(\tau) - x_i(t_0 + \tau\eta - 1))u_M}{\max(u_M,|f_r(\tau) - x_i(t_0 + \tau\eta - 1)|)} & \parbox{5em}{if input is bounded by $u_M$.}
\end{cases}
\end{align}
}
\end{lemma}

\begin{proof}
Let $\tau \geq 1$ and consider the system at time step $t = t_0 + (\tau - 1)\eta$. 
At this time step, each leader agent $l \in \Le^\N$ broadcasts the value $f_r(t_0 + \tau\eta)$ to its out-neighbors. Note that $f_r(\tau\eta)$ is the state that each behaving leader will take on at the next state update timestep; i.e. $x_l(t_0 + \tau\eta) = f_r(\tau)\ \forall l \in \Le^\N$ by \eqref{eq:leader}.
By the definition of strong $2F+1$-robustness, $|\Le| \geq F$. Since $\A$ is an $F$-local model, $|\Le^\A| \leq F$ and therefore $|\Le^\N| \geq F+1$. In addition, note that a nonempty $S_f$ implies that $\eta > |S_f| \geq 1$.

Consider the set $C_1  = \sfn = (\V \backslash \Le) \backslash \A$. Note that $C_1 \subset \V \backslash \Le$. By the definition of strong $(2F+1)$-robustness, there exists a nonempty subset $S_1 = \{i_1 \in C_1 : |\V_{i_1} \cap (\Le \cup \A)| \geq 2F+1\}$.  
Note that as per Algorithm \ref{alg:MSRPA}, the only agents which broadcast a value at time $t_0$ will be either agents in $\Le$ or agents in $\A$. This implies that agents in $S_1$ receive messages only from agents in the set $\Le \cup \A$.
Since $\A$ is $F$-local, it is impossible for any $i_1 \in S_1$ to receive a misbehaving signal from more than $F$ in-neighbors. Therefore since $|V_{i_1} \cap (\Le \cup \A)| \geq 2F+1$ and $\A$ is $F$-local, each $i_1$ receives the value $f_r(\tau)$ from $F+1$ behaving leaders.

Let $t_1 = t_0 + (\tau-1)\eta + 1$. At time $t_1$, each agent $i_1 \in S_1$ sets $u_{i_1}(t_1) = f_r(\tau) - x_{i_1}(t_1)$ if $u_{i_1}$ is unbounded, or $\frac{(f_r(\tau) - x_{i_1}(t_1))u_M}{\max(u_M,|f_r(\tau) - x_{i_1}(t_1)|)}$ if $u_{i_1}$ is bounded by $u_M$. Each agent $i_1$ then broadcasts $f_r(\tau)$ to all its out-neighbors for all time steps $t \in [t_1,t_0 + \tau\eta)$ where $t_0 + \tau\eta$ is the next state update time. Consider the set $C_2 = C_1 \backslash S_1$, which satisfies $C_2 \subseteq \V \backslash \Le$. If $C_2$ is nonempty, by the definition of strong $2F+1$-robustness there exists a nonempty subset $S_2 = \{i_2 \in C_2 : |\V_{i_2} \cap (S_1 \cup \Le \cup \A)| \geq 2F+1 \}$. As per Algorithm \ref{alg:MSRPA}, the only agents which broadcast values at time $t_1$ will be agents in the set $S_1 \cup \Le \cup \A$.
Since $\A$ is $F$-local, it is impossible for any $i_2 \in S_2$ to receive a misbehaving signal from more than $F$ in-neighbors.
This implies that each agent $i_2 \in S_2$ receives the value $f_r(\tau)$ from at least $F+1$ behaving in-neighbors in the set $S_1 \cup \Le$.

This process can be iteratively repeated: at each time step $t_p = t_0 + (\tau-1)\eta + p$, $0 \leq p < \eta$, each agent $i_p \in S_p$ will have received the value $f_r(\tau)$ from at least $F+1$ behaving in-neighbors in the previous timestep.
Each $i_p$ will therefore set $u_{i_p}(t_p) = f_r(\tau) - x_{i_p}(t_p)$ if $u_{i_p}$ is bounded, or $u_{i_p}(t_p) = \frac{(f_r(\tau) - x_{i_p}(t))u_M}{\max(u_M,|f_r(\tau) - x_{i_p}(t)|)}$ if $u_{i_p}$ is bounded by $u_M$. Each agent $i_p$ will then broadcast $f_r(\tau)$ to all its out-neighbors for all time steps $t \in [t_p,t_0 + \tau\eta)$. The set $C_{p+1}$ is iteratively defined as $C_{p+1} = C_p \backslash S_p$. By the definition of strong $(2F+1)$-robustness, if $C_{p+1}$ is nonempty there exists a nonempty $S_{p+1} = \{i_{p+1} \in C_{p+1} : |\V_{i_{p+1}} \cap \pth{\pth{\bigcup_{j=1}^p S_j} \cup \Le \cup \A}| \geq 2F+1 \}$. As per Algorithm \ref{alg:MSRPA}, the only agents which broadcast values at time $t_p$ will be agents in the set $\pth{\pth{\bigcup_{j=1}^p S_j} \cup \Le  \cup \A}$. Since $\A$ is $F$-local, it is impossible for any $i_{p+1} \in S_{p+1}$ to receive a misbehaving signal from more than $F$ in-neighbors. 
This implies that each agent $i_{p+1} \in S_{p+1}$ receives the value of $f_r(\tau)$ from at least $F+1$ behaving in-neighbors in $\pth{\pth{\bigcup_{j=1}^p S_j} \cup \Le} $.

Observe that by definition, $C_p$ being nonempty implies $C_{p+1} < C_p$. Therefore there exists a finite $p'$ such that $C_{p'}$ is empty. Since $C_p \subseteq \sfn$ for all $0 < p \leq p'$, this implies that $p' \leq |\sfn| \leq |S_f| < \eta$. 
Let $t_{p'} = t_0 + (\tau-1)\eta + p'$, and observe that $t_{p'} < t_0 + \tau\eta$.
Note that an empty $C_{p'}$ implies that $\bigcup_{k = 1}^{p'-1} S_k = \sfn$, and that all $i \in \sfn$ and all $l \in \Le^\N$ are broadcasting the value $f(\tau)$ for all timesteps $t \in [t_{p'}, t_0 + \tau\eta - 1)$.
Since $\A$ is $F$-local, this implies that no agent $i \in \sfn$ will accept any other value than $f_r(\tau)$ on time steps $t \in [t_{p'},t_0 + \tau\eta)$.
By prior arguments, equations \eqref{eq:follower} and \eqref{eq:alpha}, and step 3) in Algorithm \ref{alg:MSRPA}, agents $i_1$ in $S_1$ satisfy $x_{i_1}(t_{p'}) = x_{i_1}(t_1)$ and $u_{i_1}(t_{p'}) = f_r(\tau) - x_{i_1}(t_{p'})$ if $u_{i_1}$ is unbounded, or $\frac{(f_r(\tau) - x_{i_1}(t_{p'}))u_M}{\max(u_M,|f_r(\tau) - x_{i_1}(t_{p'})|)}$ if $u_{i_1}$ is bounded by $u_M$. Since no agent $i \in \sfn$ will accept any other value than $f_r(\tau)$ on time steps $t \in [t_{p'},t_0 + \tau\eta)$, steps 2) and 3) of Algorithm \ref{alg:MSRPA} imply that $u_{i_1}(t)$ remains constant on the interval $t \in [t_{p'},t_0 + \tau\eta)$ for all $i_1 \in S_1$.
Similarly, for $1 \leq k \leq p'-1,\ k \in \Z$, it also holds that agents $i_k$ in $S_1$ satisfy $x_{i_k}(t_{p'}) = x_{i_k}(t_k)$ and $u_{i_k}(t_{p'}) = f_r(\tau) - x_{i_k}(t_{p'})$ if $u_{i_k}$ is unbounded, or $\frac{(f_r(\tau) - x_{i_k}(t_{p'}))u_M}{\max(u_M,|f_r(\tau) - x_{i_k}(t_{p'})|)}$ if $u_{i_k}$ is bounded by $u_M$. Since no agent $i \in \sfn$ will accept any other value than $f_r(\tau)$ on time steps $t \in [t_{p'},t_0 + \tau\eta)$, steps 2) and 3) of Algorithm \ref{alg:MSRPA} imply that $u_{i_k}(t)$ remains constant on the interval $t \in [t_{p'},t_0 + \tau\eta)$ for all $i_k \in S_k$. Since it has been shown previously that $\bigcup_{k = 1}^{p'-1} S_k = \sfn$, we therefore have $u_i(t_0 + \tau\eta-1) = f_r(\tau) - x_i(t_0 + \tau\eta -1)$ if $u_i$ is unbounded, or $u_i(t_0 + \tau\eta-1) = \frac{(f_r(\tau) - x_{i}(t_{t_0 + \tau\eta -1}))u_M}{\max(u_M,|f_r(\tau) - x_{i}(t_0 + \tau\eta -1)|)}$ if $u_i$ is bounded by $u_M$ for all $i \in \sfn$.
These results hold for any $\tau \geq 1$, which concludes the proof.
\end{proof}

Using Lemma \ref{lem:broadcast}, our first Theorem demonstrates conditions under which Problem \ref{prob:leadfollow} is solved using the MS-RPA algorithm.
\begin{theorem}
\label{thm:first}
	Let $\D = (\V,\E)$ be a nonempty, nontrivial, simple digraph with $S_f$ nonempty. Let $F \in \Z_+$ and $t_0 \in \Z$. 
Define the error function $e(t) = \max_{i \in \sfn} |x_i(t) - x_l(t)|$, $l \in \Le^\N$. 
Suppose that $\A$ is an $F$-local set, $\D$ is strongly $2F+1$-robust w.r.t. the set $\Le$, and all normally behaving agents apply the MS-RPA algorithm with parameter $F$ and unbounded inputs. If $\eta > |S_f|$, then $e(t)$ is nonincreasing for $t \in [t_0,\infty)$ and $e(t) = 0$ for all $t > t_0 + \eta$. 
\end{theorem}

\begin{proof}
	By equations \eqref{eq:leader} through \eqref{eq:alpha}, each state $x_i(t)$ for all $i \in \sfn$ and $x_l(t)$ for all $l \in \Le^\N$ are constant on each time interval $t \in [t_0 + (\tau-1)\eta,t_0 + \tau\eta)$ for all $\tau \geq 1$. By Lemma \ref{lem:broadcast}, at times $t_0 + \tau\eta - 1$ for all $\tau \geq 1$ we have $u_i(t_0 + \tau\eta - 1) = f_r(\tau) - x_i(t_0 + \tau\eta - 1)$ for all $i \in \sfn$. By equations \ref{eq:leader} through \ref{eq:alpha}, we therefore have $x_i(t_0 + \tau\eta) = x_l(t_0 + \tau\eta) = f_r(\tau)$ $\forall i \in \sfn$, $\forall l \in \Le^\N$, $\forall \tau \geq 1$. It therefore holds that $e(t)$ is nonincreasing for $t \in [t_0,\infty)$ and $e(t) = 0$ for all $t \geq t_0 + \eta$.
\end{proof}




In some systems, agents' states have input bounds of the form $|u_i(t)| \leq u_M$ $\forall t \geq t_0$, $u_M \in \R_+$. Our next result shows conditions under which Problem 1 is solved by the MS-RPA algorithm even when such input bounds are imposed.
\begin{theorem}
\label{thm:two}
Let $\D = (\V,\E)$ be a nonempty, nontrivial, simple digraph with $S_f$ nonempty. Let $F \in \Z_+$, $u_M \in \R_+$, and $t_0 \in \Z$. Define the error function $e(t) = \max_{i \in \sfn} |x_i(t) - x_l(t)|$, $l \in \Le^\N$. Suppose that $\A$ is an $F$-local set, $\D$ is strongly $(2F+1)$-robust w.r.t. the set $\Le$, and all normally behaving agents apply the MS-RPA algorithm with parameter $F$. Suppose further that the input of follower agents is bounded by $|u_i(t)| \leq u_M$ $\forall t_k \geq t_0$, and the reference signal satisfies $|f_r(\tau+1) - f_r(\tau)| \leq u_M - \eps$ $\forall \tau \geq 0$, for some $\eps > 0,\ \eps \in \R$ and with $\tau \in \Z_+$. If $\eta > |S_f|$, then $e(t)$ is nonincreasing for $t \in [t_0,\infty)$ and there exists a $T \in \Z_+$ such that 
$e(t) = 0$ for all $t \geq \tau_0 + T$, $\forall i \in \sfn$, $\forall l \in \Le^\N$. Furthermore, the value of $T$ can be found as
\begin{align*}
	T  = \ceil{\frac{V(\tau_0\eta)}{ \eps}} + 1
\end{align*}
\end{theorem}

\begin{proof}
Since all behaving leaders states satisfy $|x_{l_1}(t) - x_{l_2}(t)| = 0$ $\forall t \geq t_0$ $\forall l_1,l_2 \in \Le^\N$ as per the MS-RPA algorithm and \eqref{eq:leader}, we will simply write $x_l(t)$ for brevity. 
Recall that $x_l(t) = f_r(\tau)$ $\forall t \in [(\tau-1)\eta,\tau\eta)$, $\forall \tau \geq 1$ as per \eqref{eq:leader}.
The following notation will be used for the proof:
\begin{align*}
	m(t) &= \min_{i \in \sfn} (x_i(t), x_l(t)), & M(t) &= \max_{i \in \sfn} (x_i(t), x_l(t)).	
\end{align*}
Consider the Lyapunov candidate $V(t) = M(t) - m(t)$. Clearly, $V(t) = 0$ if and only if $x_i(t) = x_l(t)$ for all $i \in \sfn$. 
Observe that by the results of Lemma \ref{lem:broadcast}, $u_i(t_0 + \tau\eta - 1) = \frac{(f_r(\tau) - x_i(t_0 + \tau\eta - 1))u_M}{\max(u_M,|f_r(\tau) - x_i(t_0 + \tau\eta - 1)|)}$ for all $i \in \sfn$, $\forall \tau \geq 1$.
The form of $u_i$ implies that for all $i \in \sfn$, the value of $x_i(t_0 + \tau\eta)$ satisfies $|x_i(t_0 + \tau\eta) - x_i(t_0 + (\tau-1)\eta)| \leq u_M$ for all $\tau \geq 1$, implying $x_i(t_0 + \tau\eta) \in B(x_i(t_0 + (\tau-1)\eta),u_M)$ (recall that $B(x,r) = \{z \in \R : |x - z| \leq r \}$). In addition, by the Theorem statement we have $|f_r(\tau) - f_r(\tau-1)| \leq u_M - \eps$ $\forall \tau \geq 1$.
We prove the result of Theorem \ref{thm:two} by showing that all of the following three statements hold true $\forall \tau \geq 1$:
\begin{enumerate}
	\item $f_r(\tau) \in B(m(t_0 + (\tau-1)\eta),u_M) \cap B(M(t_0 + (\tau-1)\eta),u_M) \implies V(t_0 + \tau\eta) = 0$	
	\item $V(t_0 + (\tau-1)\eta) \leq \eps \implies V(t_0 + \tau\eta) = 0$
	\item $V(t_0 + (\tau-1)\eta) > \eps$ and $f_r(\tau) \notin B(m(t_0 + (\tau-1)\eta),u_M) \cap B(M(t_0 + (\tau-1)\eta),u_M) \implies V(t_0 + \tau\eta) - V(t_0 +(\tau-1)\eta) \leq -\eps$
\end{enumerate}

\emph{Proof of 1):} Suppose $f_r(\tau) \in B(m(t_0 +(\tau-1)\eta),u_M) \cap B(M(t_0 +(\tau-1)\eta),u_M)$. This implies that $|f_r(\tau) - m(t_0 +(\tau-1)\eta)| \leq u_M$ and $|f_r(\tau) - M(t_0 +(\tau-1)\eta)| \leq u_M$, which implies $M(t_0 +(\tau-1)\eta) - u_M \leq f_r(\tau) \leq m(t_0 +(\tau-1)\eta) + u_M$. We have
{
\medmuskip=0mu
\thinmuskip=0mu
\begin{align*}
	f_r(\tau) &\geq M(t_0 +(\tau-1)\eta) - u_M \geq x_i(t_0 +(\tau-1)\eta) - u_M, \\
	f_r(\tau) &\leq m(t_0 +(\tau-1)\eta) + u_M \leq x_i(t_0 +(\tau-1)\eta) + u_M,
\end{align*}
}
for all $i \in \sfn$, which implies $|x_i(t_0 +(\tau-1)\eta) - f_r(\tau)| \leq u_M$ for all $i \in \sfn$. By Lemma \ref{lem:broadcast} each agent $i \in \sfn$ therefore selects $u_i(t_0 +\tau\eta - 1) = \frac{(f_r(\tau) - x_i(t_0 + (\tau-1)\eta-1 ))u_M}{\max(u_M,|f_r(\tau) - x_i(t_0 + (\tau-1)\eta - 1)|)} = f_r(\tau) - x_i(t_0 +\tau\eta-1)$, implying $|x_i(t_0 +\tau\eta) - x_l(t_0 +\tau\eta)| = 0$ $\forall i \in \sfn$ at time $t_0 +\tau\eta$, implying $V(t_0 + \tau\eta) = 0$.

\emph{Proof of 2):} $V(t_0 +(\tau-1)\eta) \leq \eps$ implies $M(t_0 +(\tau-1)\eta) - m(t_0 +(\tau-1)\eta) \leq \eps$. By definition of $m(t_0 +(\tau-1)\eta)$ and $M(t_0 +(\tau-1)\eta)$, we have $m(t_0 +(\tau-1)\eta) \leq x_l(t_0 +(\tau-1)\eta) \leq M(t_0 +(\tau-1)\eta)$. Since $|f_r(\tau) - f_r(\tau -1)| = |x_l(t_0 +\tau\eta) - x_l(t_0 +(\tau-1)\eta)| \leq u_M - \eps$ as per the Theorem statement, we have $m(t_0 +(\tau-1)\eta) - u_M + \eps \leq x_l(t_0 +\tau\eta) \leq M(t_0 +(\tau-1)\eta) + u_M - \eps$. However, $M(t_0 +(\tau-1)\eta) - m(t_0 +(\tau-1)\eta) \leq \eps$ implies $M(t_0 +(\tau-1)\eta) - u_M \leq m(t_0 +(\tau-1)\eta) - u_M + \eps$. Similarly, we have $m(t_0 +(\tau-1)\eta) \geq M(t_0 +(\tau-1)\eta) - \eps$ implying $m(t_0 +(\tau-1)\eta) + u_M \geq M(t_0 +(\tau-1)\eta) + u_M - \eps$. Therefore $M(t_0 +(\tau-1)\eta) - u_M \leq m(t_0 +(\tau-1)\eta) - u_M + \eps \leq x_l(t_0 +\tau\eta) \leq M(t_0 +(\tau-1)\eta) + u_M - \eps \leq m(t_0 +(\tau-1)\eta) + u_M$, which implies $f_r(\tau) = x_l(t_0 +\tau\eta)$ must satisfy $f_r(\tau) = x_l(t_0 +\tau\eta) \in B(m(t_0 +(\tau-1)\eta),u_M) \cap B(M(t_0 +(\tau-1)\eta),u_M)$, implying $V(t_0 +\tau\eta) = 0$ from the arguments in the Proof of 1).

\emph{Proof of 3):} $V(t_0 + (\tau-1)\eta) > \eps$ and $f_r(\tau) \notin B(m(t_0 + (\tau-1)\eta),u_M) \cap B(M(t_0 + (\tau-1)\eta),u_M)$ imply that $x_l(t_0 + \tau\eta) = f_r(\tau) \notin B(m(t_0 + (\tau-1)\eta),u_M)$ or $x_l(t_0 + \tau\eta) = f_r(\tau) \notin B(M(t_0 + (\tau-1)\eta),u_M)$. Without loss of generality, consider the case where $x_l(t_0 + \tau\eta) \notin B(M(t_0 + (\tau-1)\eta),u_M)$. Since $m(t_0 + (\tau-1)\eta) \leq x_l(t_0 + (\tau-1)\eta)$ by definition, and $|x_l(t_0 + \tau\eta) - x_l(t_0 + (\tau-1)\eta)| \leq u_M - \eps$, we therefore have $m(t_0 + (\tau-1)\eta) - x_l(t_0 + \tau\eta) \leq u_M - \eps$, which implies $f_r(\tau) = x_l(t_0 + \tau\eta) \geq m(t_0 + (\tau-1)\eta)-(u_M - \eps)$. 
Since $u_i(t_0 + \tau\eta-1) = \frac{(f_r(\tau) - x_i(t_0 + (\tau-1)\eta-1 ))u_M}{\max(u_M,|f_r(\tau) - x_i(t_0 + (\tau-1)\eta - 1)|)}$ $\forall i \in \sfn$ by Lemma \ref{lem:broadcast}, $f_r(\tau) \geq m(t_0 + (\tau-1)\eta)$ implies $m(t_0 + \tau\eta) \geq m(t_0 + (\tau-1)\eta)$ and $f_r(\tau) < m(t_0 + (\tau-1)\eta)$ implies $m(t_0 + \tau\eta) \geq m(t_0 + (\tau-1)\eta)-(u_M - \eps)$.
We therefore have $m(t_0 + (\tau-1)\eta) -  m(t_0 + \tau\eta)  \leq (u_M - \eps)$.
By Lemma \ref{lem:broadcast} all agents $i \in \sfn$ such that $x_i(t_0 + (\tau-1)\eta) = M(t_0 + (\tau-1)\eta)$ will have $u_i(t_0 + \tau\eta-1) = \frac{(f_r(\tau) - x_i(t_0 + (\tau-1)\eta-1 ))u_M}{\max(u_M,|f_r(\tau) - x_i(t_0 + (\tau-1)\eta - 1)|)} = u_M \text{sgn}(f_r(\tau) - x_i(t_0 + (\tau-1)\eta -1)) = -u_M$, since $x_l(t_0 + \tau\eta) = f_r(\tau) \notin B(M(t_0 + (\tau-1)\eta),u_M)$. This implies $M(t_0 + \tau\eta) - M(t_0 + (\tau-1)\eta) = -u_M$. 
Therefore,
{
\medmuskip=0mu
\thinmuskip=0mu
\thickmuskip=0mu
\begin{align*}
	&V(t_0 + \tau\eta) - V(t_0 + (\tau-1)\eta) = M(t_0 + \tau\eta) \nonumber\\ & \hspace{4em}- m(t_0 + \tau\eta) 
	  - (M(t_0 + (\tau-1)\eta) - m(t_0 + (\tau-1)\eta)), \nonumber\\
	&= M(t_0 + \tau\eta) - M(t_0 + (\tau-1)\eta) + m(t_0 +(\tau-1)\eta) - m(t_0 + \tau\eta), \nonumber\\
	&\leq  -u_M + u_M - \eps \hspace{.5em} \leq \hspace{.5em} -\eps.
\end{align*}
}
Similar arguments hold if $x_l(t_0 + \tau\eta) \notin B(m(t_0 + (\tau-1)\eta),u_M)$, which yields 3).

From 1), 2) and 3) we can infer that if $V(t_0 +(\tau-1)\eta) > \eps$ and $x_l(t_0 +\tau\eta) \notin B(m(t_0 +(\tau-1)\eta),u_M) \cap B(M(t_0 +(\tau-1)\eta),u_M)$ for $\tau \geq 1$, then $V(t_0 +\tau\eta) - V(t_0 +(\tau-1)\eta) \leq -\eps$. This implies that the value of $V(\cdot)$ decreases until $V(t_0 +(\tau-1+k')\eta) \leq \eps$ or $x_l(t_0 +(\tau+k')\eta) \in B(m(t_0 +(\tau-1+k')\eta),u_M) \cap B(M(t_0 +(\tau-1+k')\eta),u_M)$ for some finite $k' \geq 1$, either of which conditions imply that $V(t_0 +(\tau+k'+1)\eta) = 0$. In addition, note that $V(t_0 +(\tau+k')\eta) = 0 \implies V(t_0 +(\tau+k')\eta) \leq \eps$, implying by 2) that $V(t_0 +(\tau+k'+p)\eta) = 0$ for all $p \geq 0,\ p \in \Z_+$. 
We can therefore conclude that there exists a $T \in \Z_+$ such that $|x_i(t) - x_l(t)| = 0$ $\forall i \in \sfn$, $\forall t \geq \tau_0 + T$, which implies $e(t) = 0$ $\forall t \geq \tau_0 + T$. In addition, statements 1),  2), and 3) imply that $(V(t_0 + \tau\eta) - V(t_0 + (\tau-1)\eta) \leq 0$ $\forall \tau \geq 1$, which along with \eqref{eq:leader}, \eqref{eq:follower}, and \eqref{eq:alpha} implies that $e(t)$ is nonincreasing $\forall t \geq t_0$.

An exact value of $T$ can be found by using statement 3) which implies that the slowest rate of decrease of $V(\cdot)$ is $V(t_0 +\tau\eta) - V(t_0 +(\tau-1)\eta) \leq -\eps$. Letting $k' = (\ceil{V(t_0) / \eps} + 1)$, it is straighforward to show that after $k'\eta$ time steps, $V(t_0 + k'\eta) = 0$. Therefore $T$ can be taken to be $T = k' = (\ceil{V(t_0) / \eps} + 1)$.
\end{proof}

}

\section{Simulations}
\label{sec:simulations}

This section presents simulations of the MS-RPA algorithm conducted on an undirected $k$-circulant digraph \cite{usevitch2018resilient,usevitch2017r} denoted $\D_1$ with $n = 14$ agents and parameter $k = 5$. The set of leaders is $\Le = \{1,\ldots,5\}$, and $S_f = \{6,\ldots,14\}$. Our results in \cite{usevitch2018resilient} demonstrate conditions under which $k$-circulant undirected and directed graphs are strongly $r$-robust with respect to $\Le$. Briefly, if a $k$-circulant undirected or directed graph $\D$ contains a set of consecutive agents by index $P_L$ such that $|P_L| \leq k$ and $|P_L \cap \Le| \geq r$, $\D$ is strongly $r$-robust w.r.t. $\Le$. Letting $P_L = \{1,\ldots,5 \}$, we clearly have $|P_L| \leq k = 5$ and $|P_L \cap \Le| \geq 5$, implying $\D_1$ is $5$-robust. The parameter $F$ is set to $F=2$, implying that $\D_1$ is strongly $(2F+1)$-robust w.r.t. $\Le$. 

In the first simulation, agents inputs are unbounded. All agents begin with randomly initial states on the interval $[-25,25]$. The initial time $t_0 = 0$, the reference function is $f_r(\tau) = 10\sin(\tau/\pi)$, and the communication rate is defined by $\eta = 10$. Note that $\eta > |S_f| = 9$.
Two leader agents $\A = \{1, 5\}$ are misbehaving by updating their state according to an arbitrary function \emph{and} sending random values to their out-neighbors at each time step;  the misbehavior is \emph{malicious} \cite{leblanc2013resilient} because each misbehaving agent sends the same misinformation to each of its out-neighbors. We emphasize that \emph{the normally behaving agents have no knowledge about which agents are misbehaving}. All other normally behaving agents apply the MS-RPA algorithm (Algorithm \ref{alg:MSRPA}). 
The normally behaving follower agents achieve consensus to the leader agents at time step $t_0 + \eta$, where $t_0 = 0$, and successfully track the leaders' time-varying state \emph{exactly} for all remaining time. 

In the second simulation, all parameters are the same  as in the first simulation except 1) agents' inputs are bounded by $u_M = 10.1$, and 2) the set of misbehaving agents is $A = \{4,11\}$ (one leader, one follower). Again, normally behaving agents have no knowledge about which agents are misbehaving. Note that $|f_r(\tau+1) - f_r(\tau)| < u_M$ for all $\tau \in Z_+$, and therefore convergence is guaranteed after a finite number of time steps $T$ as per Theorem \ref{thm:two}. This is shown in Figure \ref{fig:finite}, where the followers converge exactly to the leaders after a finite number of time steps despite the influence of the misbehaving agents.

\begin{figure}
\centering
\includegraphics[width=.88\columnwidth]{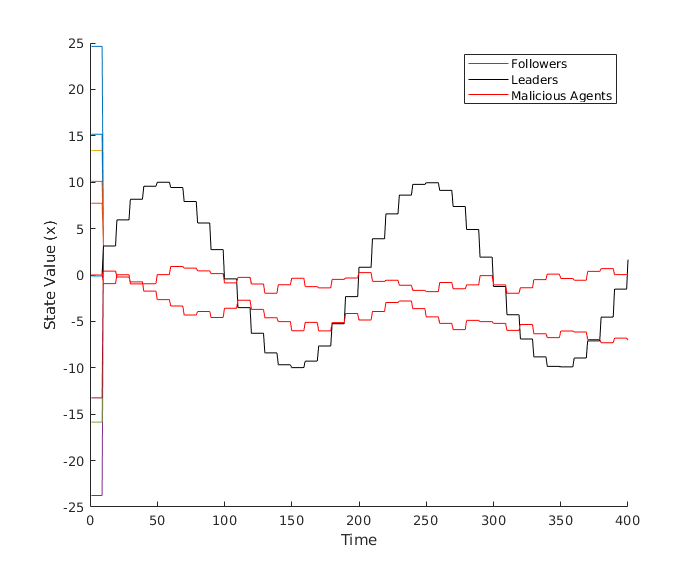}
\caption{Demonstration of the MS-RPA algorithm. Leader agents' states are denoted by black lines, misbehaving agent's states are denoted by red lines, and follower agents' states are denoted by the multi-colored lines. When the inputs are unbounded, exact convergence to the leaders' trajectory occurs for all $t \geq \eta$, where $\eta=10$ is the communication rate.}
\label{fig:unbounded}
\end{figure}

\begin{figure}
\centering
\includegraphics[width=.88\columnwidth]{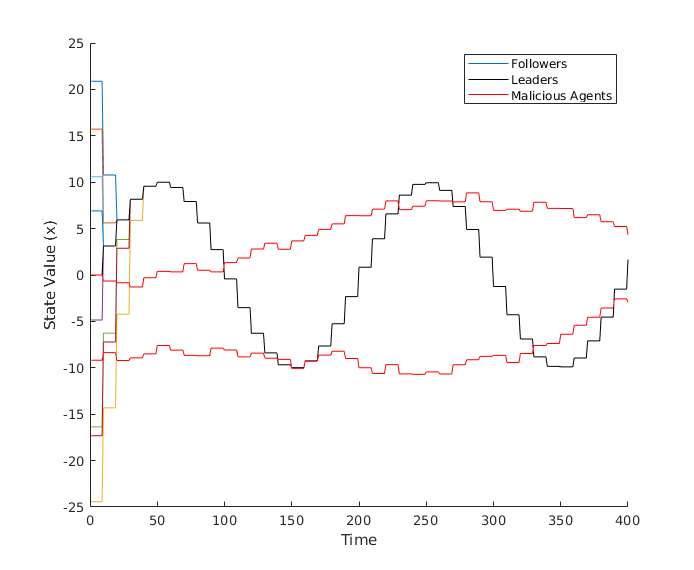}
\caption{Demonstration of the MS-RPA algorithm with input bounds. When the inputs of the agents are bounded, exact convergence to the time-varying leaders' states is still guaranteed according to the conditions in Theorem \ref{thm:two}. In this simulation, we again have $\eta = 10$.}
\label{fig:finite}
\end{figure}

\section{Conclusion}
\label{sec:conclusion}

In this paper, we presented algorithms and conditions for agents with discrete-time dynamics to resiliently track a reference signal propagated by a set of leader agents despite a bounded number of the leaders and followers behaving adversarially. Future work will include extending these results to time-varying graphs with asynchronous communication.

\bibliographystyle{IEEEtran}

\bibliography{technote.bib}

\end{document}